\documentclass[12pt]{llncs}

\usepackage{amsmath,amssymb}

\usepackage[all]{xy}
\usepackage{latexsym}

\usepackage{enumitem}
\usepackage{algorithmic}
\usepackage{algorithm}

\makeatletter
\newcommand{\todo}[1]{\marginpar{\textbf{TODO\footnotemark}}\@latex@warning{TODO: #1}\footnotetext{ #1}}
\makeatother

\newcommand{\up}[1]{[#1]}
\newcommand{\atomP}{\mathcal{AP}}
\newcommand{\atomQ}{\mathcal{AQ}}

\newcommand{\langCl}{\mathcal{L}_\mathsf{cl}}
\newcommand{\langB}{\mathcal{L}_\mathsf{B}}
\newcommand{\langM}{\mathcal{L}_\mathsf{M}}

\newcommand{\natN}{\mathbb{N}}
\newcommand{\natNP}{\mathbb{N}^+}
\newcommand{\seq}[1]{\langle#1\rangle}
\newcommand{\len}{\mathsf{len}}
\newcommand{\set}[1]{\mathsf{set}(#1)}
\newcommand{\apnd}{\circ}
\newcommand{\M}{\mathsf{M}}
\newcommand{\K}{\mathsf{K}}
\newcommand{\I}{\mathsf{I}}
\newcommand{\BC}{\mathsf{BC}}
\newcommand{\PU}{\mathsf{PU}}
\newcommand{\val}{\mathsf{v}}
\newcommand{\complete}{\mathsf{complete}}
\newcommand{\find}{\mathsf{find}}
\newcommand{\remove}{\mathsf{remove}}
\newcommand{\fml}{\mathsf{fml}}
\newcommand{\ind}{\mathsf{ind}}
\newcommand{\BCL}{\mathsf{BCL}}
\newcommand{\Acc}{\mathsf{Acc}}
\newcommand{\modelsCL}{\models_\mathsf{CL}}
\newcommand{\modelsM}{\models_\Box}
\newcommand{\tr}{\mathsf{t}}
\newcommand{\h}{\mathsf{h}}

\begin{document}
\frontmatter          
\mainmatter              
\title{A Logic of Blockchain Updates}
\titlerunning{Epistemic Blockchain Logic}  
%
\author{Kai Br\"unnler\inst{1} \and Dandolo Flumini\inst{2} \and Thomas Studer\inst{3}}
\authorrunning{Kai Br\"unnler et al.} 
%
%
\institute{
Bern University of Applied Sciences, Switzerland, 
\email{kai.bruennler@bfh.ch}
\and
ZHAW School of Engineering, Switzerland, 
\email{dandolo.flumini@zhaw.ch}
\and
University of Bern, Switzerland,
\email{tstuder@inf.unibe.ch}
}

\maketitle              

\begin{abstract}
Blockchains are distributed data structures that are used to achieve consensus in systems for cryptocurrencies (like Bitcoin) or smart contracts (like Ethereum). 
Although blockchains gained a lot of popularity recently, there is no logic-based model for blockchains available.
We introduce~$\BCL$, a dynamic logic to reason about blockchain updates, and show that $\BCL$ is sound and complete with respect to a simple blockchain model.
\end{abstract}
\section{Introduction}

Bitcoin~\cite{nakamoto}
 is  a cryptocurrency that uses peer-to-peer technology to support direct user-to-user transactions without an intermediary such as a bank or credit card company. In order to prevent double spending, which is a common issue in systems without central control, Bitcoin maintains a complete and public record of all transactions at each node in the network.
This ledger is called the \emph{blockchain}. 

The blockchain is essentially a growing sequence of blocks, which contain approved transactions and a cryptographic hash of the previous block in the sequence. Because the blockchain is stored locally at each node, any update to it has to be propagated to the entire network. Nodes that receive a transaction first verify its validity (i.e., whether it is compatible with all preceeding transactions). If it is valid, then it is added to the blockchain and sent to all other nodes~\cite{Antonopoulos:2014,Sompolinsky2015}. Blockchain technology, as a general solution to the Byzantine Generals' Problem~\cite{Lamport:1982}, is now not only used for financial transactions but also for many other applications like, e.g., smart contracts~\cite{Buterin2013}.

Herlihy and Moir~\cite{Herlihy:2016} propose to develop a logic of accountability to design and verify blockchain systems. In particular, they discuss blockchain scenarios to test (i) logics of authorization, (ii) logics of concurrency, and (iii) logics of incentives.

In the present paper, we are not interested in accountability but study blockchains from the perspective of dynamic epistemic logic~\cite{ditmarsch:07}. 
A given state of the blockchain entails knowledge about the transactions that have taken place. We ask:
\emph{how does this knowledge change when a new block is received that might be added to the blockchain?}
We develop a dynamic logic, $\BCL$, with a semantics that is based on a blockchain model. 
The update operators of $\BCL$ are interpreted as receiving new blocks. 
It is the aim of this paper to investigate the dynamics 
of blockchain updates.

The deductive system for $\BCL$ includes reduction axioms that make it possible to establish completeness by a reduction to the update-free case~\cite{kooi}. However, since blockchain updates are only performed if certain consistency conditions are satisfied, we use conditional reduction axioms similar to the ones developed by Steiner to model consistency preserving updates~\cite{steiner}.
Moreover, unlike traditional public announcements~\cite{ditmarsch:07}, blockchain updates cannot lead to an inconsistent state, i.e.,  updates are total, like in~\cite{SteinerStuder}.

We do not base $\BCL$ on an existing blockchain implementation but use a very simple model. First of all, the blockchain is a sequence of propositional formulas. Further we  maintain a list of provisional updates. Our blocks consist of two parts: a sequence number (called the index of the block) and a propositional formula.
If a block is received, then the following case distinction is performed where $i$ is the index of the block and $l$ is the current length of the blockchain:
\begin{enumerate}
\item
$i \leq l$. The block is ignored.
\item 
$i = l+1$. If the formula of the block is consistent with the blockchain, then it is added to the blockchain; otherwise the block is ignored. If the blockchain has been extended, then this procedure is performed also with the blocks stored in the list of provisional updates.
\item
$i > l+1$. The block is added to the list of provisional updates.
\end{enumerate}
Although this is a simple model, it features two important logical properties of blockchains: consistency must be preserved and blocks may be received in the wrong order in which case they are stored separately until the missing blocks have been received. 

The main contribution of our paper from the point of view of dynamic epistemic logic is that we maintain a list of provisional updates. That means we support updates that do not have an immediate effect but that may lead to a belief change later only after certain other updates have been performed. $\BCL$ is the first dynamic epistemic logic that features provisional updates of this kind.

The paper is organized as follows.
The next section introduces our blockchain model, the language of $\BCL$, and its semantics.
In Section~\ref{sec:3}, we introduce a deductive system for $\BCL$.
We establish soundness of $\BCL$ in Section~\ref{sec:4}.
In Section~\ref{sec:5}, we show a normal form theorem for $\BCL$, which is used in Section~\ref{sec:6} to prove completeness of $\BCL$.
The final section studies some key principles of the epistemic dynamics
of our blockchain logic and discusses future work.




%
%
%
%
%
%
%
%
%
%
%
%
%
%

\section{A simple dynamic epistemic blockchain logic}

The set of all natural numbers is denoted by $\natN := \{0,1,2,\ldots\}$.
The set of positive natural numbers is denoted by $\natNP := \{1,2,\ldots\}$. We use $\omega$ for the least ordinal such that $\omega >n$, for all $n \in \natN$.

Let $\sigma = \seq{\sigma_1,\ldots,\sigma_n}$ be a finite sequence.
We define its \emph{length}\/  by $\len(\sigma) := n$.
For an infinite sequence $\sigma = \seq{\sigma_1,\sigma_2,\ldots}$ we set 
$\len(\sigma) := \omega$.
Further for a (finite or infinite) sequence $\sigma =  \seq{\sigma_1, \sigma_2, \ldots,\sigma_i,\ldots}$ we set $(\sigma)_i:= \sigma_i$.
The \emph{empty sequence}\/ is denoted by $\seq{}$ and we set 
$\len(\seq{}) := 0$.
We can append $x$ to a finite sequence $\sigma := \seq{\sigma_1,\ldots,\sigma_n}$, in symbols we set $\sigma \apnd x := \seq{\sigma_1,\ldots,\sigma_n,x}$.
We will also need the set of all components of a sequence $\sigma$ and define
\[
\set{\sigma} := \{ x \ |\ \text{there is an $i$ such that $x = \sigma_i$} \}.
\] 
In particular, we have $\set{\seq{}} := \emptyset$.
Moreover, we use the shorthand $x \in \sigma$ for  $x \in \set{\sigma}$.

We start with a countable set of atomic propositions~$\atomP:=\{P0, P1, \ldots\}$.
The set of formulas~$\langCl$ of classical propositional logic is given by the following grammar
\[
A ::= \bot \ |\  P \ |\ A \to A \quad,
\]
where $P \in \atomP$. 

In order to introduce the language~$\langB$ for blockchain logic, we need another countable set of special atomic propositions $\atomQ:=\{Q1, Q2, \ldots\}$ that is disjoint with $\atomP$.
We will use these special propositions later to keep track of the length of the blockchain.
The formulas of $\langB$  are now
given by the grammar
\[
F ::= \bot \ |\  P \ |\  Q \ |\  F \to F \ |\ \Box A \ |\ \up{i,A} F \quad,
\]
where  $P \in \atomP$, $Q \in \atomQ$, $A \in \langCl$, and $i \in \natNP$. The operators of the form $\up{i,A}$ are called \emph{blockchain updates} (or simply \emph{updates}).

Note that in $\langB$ we cannot express higher-order knowledge, i.e., we can only express knowledge about propositional facts but not knowledge about knowledge of such facts.

For all languages in this paper, we define further Boolean connectives (e.g.~for negation, conjunction, and disjunction) as usual. Moreover, we assume that unary connectives bind stronger than binary ones.
 
For $\langCl$ we use the semantics of classical propositional logic.
A \emph{valuation} $\val$ is a subset of  $\atomP$ and we define the truth of an $\langCl$-formula~$A$ under $\val$, in symbols $\val\models A$ as usual.
For a set $\Gamma$ of $\langCl$-formulas, we write $\val\models \Gamma$ if $\val\models A$ for all $A \in \Gamma$.
The set $\Gamma$ is \emph{satisfiable} if there is a valuation $\val$ such that $\val\models \Gamma$. We say $\Gamma$ \emph{entails} $A$, in symbols $\Gamma \models A$, if for each valuation $\val$ we have
\[
\val \models \Gamma \quad\text{implies}\quad \val \models A.
\]
Now we introduce the blockchain semantics for~$\langB$.
\begin{definition}
A \emph{block} is a pair $[i,A]$ where $A$ is an $\langCl$-formula and $i \in \natNP$. We call $i$ the \emph{index} and $A$ the \emph{formula} of the block $[i,A]$.
We define functions~$\ind$ and~$\fml$ by $\ind[i,A]:=i$ and $\fml[i,A]:=A$.
\end{definition}

\begin{definition}
A \emph{model} $\M := (\I, \BC, \PU, \val)$ is a quadruple where 
\begin{enumerate}
\item $\I$ is a set of $\langCl$-formulas
\item $\BC$ is a  sequence of $\langCl$-formulas
\item $\PU$ is a finite sequence of blocks
\item $\val$ is a valuation, i.e.~$\val \subseteq \atomP$
\end{enumerate}
 such that
\begin{equation}
\label{eq:model:1}
\text{$\I \cup \set{\BC}$ is satisfiable}
\end{equation}
and
\begin{equation}
\label{eq:model:2}
\text{for each block $[i,A] \in \PU$ we have $i> \len(\BC) +1$.}
\end{equation}
\end{definition}
The components of a model $ (\I, \BC, \PU, \val)$ have the following meaning:
\begin{enumerate}
\item $\I$ models initial background knowledge.
\item $\BC$ is the blockchain. 
\item $\PU$ stands for \emph{provisional updates}. The sequence $\PU$ consists of those blocks that have been announced but that could not yet be added to the blockchain because their index is too high. Maybe they will be added to $\BC$ later (i.e., after the missing blocks have been added). 
\item $\val$ states which atomic propositions are true.
\end{enumerate}

We need some auxiliary definition in order to precisely describe the blockchain dynamics.

\begin{definition}
\begin{enumerate}
\item Let\/ $\PU$ be a finite sequence of blocks. Then we let $\find(i,\PU)$ be the least~$j \in \natNP$ such that there is an $\langCl$-formula $A$ with $[i,A] = (\PU)_j$.
\item Let $\sigma = \seq{\sigma_1,\ldots,\sigma_{i-1},\sigma_{i},\sigma_{i+1},\ldots}$ be a  sequence. We set 
\[
\remove(i,\sigma) := \seq{\sigma_1,\ldots,\sigma_{i-1},\sigma_{i+1},\ldots}.
\]
\item Given a set of $\langCl$-formulas~$\I$, a sequence of\/~$\langCl$-formulas\/~$\BC$, and  
a finite sequence of blocks\/~$\PU$,  then the \emph{chain completion} $\complete(\I, \BC,\PU)$ is computed according to Algorithm~\ref{Alg:ChainCompletion}.
\end{enumerate}
\end{definition}

\begin{algorithm}[H]
	\caption{Chain Completion Algorithm: $\complete$ }
	\label{Alg:ChainCompletion}
	\begin{algorithmic}[1] 
		\REQUIRE $(\I, \BC,\PU)$
		\STATE{$n \gets \len(\BC)+1$}
		\WHILE{ $[n,A] \in \PU$ for some formula $A$}
			\STATE{$i \gets \find(n,\PU)$}
			\STATE{$B \gets \fml((\PU)_i)$}
			\STATE{$\remove(i,\PU)$}
			\IF{$\I \cup \set{\BC} \cup\{B\}$ is satisfiable}
				\STATE{$\BC \gets \BC \apnd B$}
				\STATE{$n \gets \len(\BC)+1$}
			\ENDIF
		\ENDWHILE
		\FOR {$i \in \len(\PU),\ldots,1$ }
			\IF{$\ind((\PU)_i) < n$}	 
				\STATE{$\remove(i,\PU)$}     
			\ENDIF
		\ENDFOR
		\RETURN $(\BC,\PU)$
	\end{algorithmic}
\end{algorithm}

Let us comment on the chain completion procedure. The numbers refer to the lines in Algorithm~\ref{Alg:ChainCompletion}.
\begin{itemize}[leftmargin=4em]
\item[1:] $n$ is the index a block must contain so that it could be added to the blockchain~$\BC$.
\item[2:] '$[n,A] \in \PU$ for some formula $A$' means that $\PU$ contains a block that could be added to $\BC$.
\item [3--5:] Find the next formula $B$ that could be added to $\BC$ and remove the corresponding block from $\PU$.
\item [6:] '$\I \cup \set{\BC} \cup\{B\}$ is satisfiable' means that $B$ is consistent with the current belief. This test guarantees that~$\eqref{eq:model:1}$ will always be satisfied.
\item [7,8:] Update the blockchain~$\BC$ with $B$.  
\item [11--15:] Remove all blocks from $\PU$ whose index is less than or equal to the current length of the blockchain~$\BC$. Because the blockchain never gets shorter, these block will never be added. Removing them guarantees that~$\eqref{eq:model:2}$ will always be satisfied.
\end{itemize}

Note if $\BC$ and $\PU$ satisfy condition~\eqref{eq:model:2} in the definition of a model, then the chain completion algorithm  will return $\BC$ and $\PU$ unchanged.

\begin{lemma}\label{l:welldef:1}
Let $\I$ be a set of  $\langCl$-formulas and
let  $\BC$ be a sequence of $\langCl$-formulas such that\/ $\I \cup \set{\BC}$ is satisfiable. Let $\PU$ be an arbitrary finite sequence of blocks.
For $(\BC',\PU'):= \complete(\I,\BC,\PU)$ we find that
\begin{enumerate}
\item $\I \cup \set{\BC'}$ is satisfiable and
\item for each block $[i,A] \in \PU'$ we have $i> \len(\BC') +1$.
\end{enumerate}
\end{lemma}
\begin{proof}
By assumption,
\begin{equation}\label{eq:invariant:1}
\text{$\I \cup \set{\BC}$ is satisfiable}
\end{equation}
holds for the arguments passed to the algorithm. Moreover, the condition in line~6 guarantees that~\eqref{eq:invariant:1}
is a loop invariant of the while loop in lines~2--10, i.e., it holds after each iteration. Since $\BC$ is not changed after line 10, \eqref{eq:invariant:1} also holds for the final result, which shows the first claim of the lemma.

It is easy to see that
\begin{equation}\label{eq:invariant:2}
n = \len(\BC)+1
\end{equation}
also is a loop invariant of while loop in lines~2--10.
In particular, \eqref{eq:invariant:2} holds after line 10 and thus the for loop in lines~11--15 removes all blocks $[i,A]$ from $\PU$ with $ i < \len(\BC)+1$.
Moreover, after the while loop in lines~2--10 has terminated,  its loop condition must be false, which means that $\PU$ cannot contain a block 
$[i,A]$ with $i = \len(\BC)+1$.
This finishes the proof of the second claim.
\qed
\end{proof}

\begin{definition}
Let $\M := (\I, \BC, \PU, \val)$  be a model and\/ $[i,A]$ be a block.
The \emph{updated model} $\M^{[i,A]}$ is  defined as  $(\I, \BC', \PU', \val)$ where 
\[
(\BC',\PU') := \complete(\I, \BC,\PU \apnd [i,A] ).
\]
\end{definition}

\begin{remark}
Note that $\M^{[i,A]}$ is well-defined: by Lemma~\ref{l:welldef:1} we know that $\M^{[i,A]}$  is indeed a model.
\end{remark}

\begin{definition}
\label{d:truth:1}
Let $\M := (\I, \BC, \PU, \val)$  be a model. We define the \emph{truth} of an $\langB$-formula $F$ in $\M$, in symbols $\M \models F$, inductively by:
\begin{enumerate}
\item $\M \not \models \bot$;
\item $\M \models P$ if $P \in \val$ for $P\in \atomP$;
\item $\M \models Qi$ if $i \leq \len(\BC)$ for  $Qi\in \atomQ$;
\item $\M \models F \to G$ if\/ $\M \not \models F$ or\/ $\M \models G$;
\item $\M \models \Box A$ if\/ $\I \cup \set{\BC}\models A$;
\item $\M \models [i,A] F$ if\/ $\M^{[i,A]} \models F$.
\end{enumerate}
\end{definition}

We define validity only with respect to the class of models that do not have provisional updates. 
\begin{definition}
We call a model\/ $\M = (\I, \BC, \PU, \val)$ \emph{initial} if\/ $\PU =\seq{}$.
A formula~$F$ is called \emph{valid} if\/ $\M \models F$ for all initial models~$\M$.
\end{definition}



\section{The deductive system $\BCL$}\label{sec:3}


In order to present an axiomatic system for our blockchain logic, we need to formalize an \emph{acceptance condition} stating whether a received block can be added to the blockchain. That is we need a formula~$\Acc(i,A)$ expressing that the formula~$A$ is consistent with the current beliefs and the current length of the blockchain is $i-1$. 
Thus if  $\Acc(i,A)$ holds, then the block $[i,A]$ 
will be accepted and added to the blockchain.
The truth definition for the atomic propositions $Qi \in \atomQ$ says that $Qi$ is true if the blockchain contains at least $i$ elements. That means the formula $Q(i-1) \land \lnot Qi$ is true if the blockchain contains exactly $i-1$ elements. This leads to the following definition of $\Acc(i,A)$ for  $i \in \natNP$:
%
\[
\Acc(i,A) := \begin{cases}
\lnot Qi \land \lnot \Box \lnot A & \text{ if $i=1$}\\
Q(i-1) \land \lnot Qi \land \lnot \Box \lnot A & \text{ if $i>1$}
\end{cases}
\]
%
As desired, we find that  if $\Acc(i,A)$ is true, then the chain completion algorithm can append the formula~$A$ to the blockchain (see Lemma~\ref{l:updates:1} later).

An $\langB$-formula is called compliant if the blockchain updates occur in the correct order. Formally, we use the following definition.
\begin{definition}
An $\langB$-formula $F$ is \emph{compliant} if
no occurrence of a\/ $[i,A]$-operator in~$F$ is in the scope of some $[j,B]$-operator with $j > i$.
\end{definition}

Now we can define the system~$\BCL$ for \emph{Epistemic Blockchain Logic}.
It is formulated in the language~$\langB$ and consists of the following axioms:
\par\medskip
\begin{tabular}{ll}
$( \mathsf{PT} )$ & Every instance of a
propositional tautology \\
$( \mathsf{K} )$ & $\Box ( F \rightarrow G ) \rightarrow ( \Box F
\rightarrow \Box G )$\\
$( \mathsf{D} )$ & $\lnot \Box \bot$\\
$( \mathsf{Q} )$ & $Qi \to Qj$ if $i > j$\\
$( \mathsf{A1} )$ & $ [i,A] \bot \to \bot$\\
$( \mathsf{A2} )$ & $ [i,A]P \leftrightarrow P$ for $P\in\atomP$\\
$( \mathsf{A3.1} )$ & $ \Acc(i,A) \to ( [i,A]Qi \leftrightarrow \top)$ for $Qi \in\atomQ$\\
$( \mathsf{A3.2} )$ & $ \lnot \Acc(i,A) \to( [i,A]Qi \leftrightarrow Qi)$ for $Qi \in\atomQ$\\
$( \mathsf{A3.3} )$\quad\   & $ [i,A]Qj \leftrightarrow Qj$ for $Qj \in\atomQ$ and $i\neq j$\\
$( \mathsf{A4} )$ & $\begin{array}{l}  [i_1,A_1]\ldots[i_k,A_k](F \to G) \leftrightarrow\\
\qquad\qquad ( [i_1,A_1]\ldots[i_k,A_k]F \to [i_1,A_1]\ldots[i_k,A_k]G)\end{array}$ \\
$( \mathsf{A5.1} )$ & $ \Acc(i,A) \to ([i,A] \Box B  \leftrightarrow \Box(A \to B))$\\
$( \mathsf{A5.2} )$ & $ \lnot \Acc(i,A) \to ([i,A] \Box B  \leftrightarrow \Box B)$\\  
$( \mathsf{A6} )$ & $\begin{array}{l}  [h_1,C_1]\ldots[h_k,C_k][i,A][j,B] F  \leftrightarrow  \\ \qquad\qquad[h_1,C_1]\ldots[h_k,C_k][j,A][i,B] F \end{array}$ \quad for $i \neq j$
\end{tabular}
\par\medskip

Note that in $( \mathsf{A6} )$, we may choose $k$ to be 0, in which case the axiom has the form $[i,A][j,B] F  \leftrightarrow [j,A][i,B] F$ for $i \neq j$.


%
In order to formulate the rules of $\BCL$, we need the following notation.
Let $H(P)$ be a formula that may contain occurrences of the atomic proposition $P$. By $H(F)$, we denote the result of simultaneously replacing each occurrence of $P$ in $H(P)$ with the formula $F$.
The rules of $\BCL$ are:
\[
( \mathsf{MP})\, \begin{array}{c} F \qquad F \rightarrow G \\
\hline G \end{array} 
\qquad 
( \mathsf{NEC} )\, \begin{array}{c} A
\\ \hline \Box A \end{array} 
%
\qquad 
( \mathsf{SUB} )\, \begin{array}{c} F \leftrightarrow G \\
\hline H(F) \leftrightarrow H(G) \end{array} 
\]
where $\mathsf{(SUB)}$ can only be applied if 
$H(F) \leftrightarrow H(G)$ is a compliant formula.

%
%

\begin{remark}
Our semantics includes infinite blockchains: in a given model~$(\I, \BC, \PU, \val)$, the sequence $\BC$ may have infinite length. If we want to exclude such models, then we have to add an infinitary rule 
\[
\begin{array}{c} Qi \quad \text{for all $i\in \natNP$} \\
\hline \bot \end{array} 
\
\]
to $\BCL$. This rule states that some $Qi$ must be false, which means that $\BC$ has finite length.
\end{remark}

\section{Soundness}\label{sec:4}

Before we can establish soundness of $\BCL$, we have to show some preparatory lemmas.

\begin{lemma}\label{l:updates:1}
Let\/ $\M := (\I, \BC, \seq{}, \val)$ be an initial model. Further let $(\I, \BC', \PU', \val) := \M^{[i,A]}$ for some block $[i,A]$.
\begin{enumerate}
\item
If\/ $\M \models \Acc(i,A)$, then $\BC' = \BC \apnd A$.
In particular, this yields
$\len(\BC') = i$ and
for each $j$ with $j \neq i$,
\[
M\models Qj \quad\text{if and only if}\quad \M^{[i,A]} \models Qj.
\]
\item
If\/ $\M \not\models \Acc(i,A)$, then $\BC' = \BC$. 
\end{enumerate}
\end{lemma}
\begin{proof}
Assume $\M \models \Acc(i,A)$. That means $\len(\BC)+1 = i$ and
$\I \cup \set{\BC} \cup \{A\}$ is satisfiable. Hence we find 
\[
\complete(\I, \BC,\seq{} \apnd [i,A]) = (\BC \apnd A, \seq{}). 
\]
Therefore $\BC' = \BC \apnd A$.
This immediately yields 
\[
\len(\BC') =i = \len(\BC)+1
\]
and for each $j$ with $j \neq i$,
\[
M\models Qj \quad\text{if and only if}\quad \M^{[i,A]} \models Qj.
\]

Assume $\M \not\models \Acc(i,A)$. This implies 
\[
\text{$\len(\BC)+1\neq i$ or  $\I \cup \set{\BC} \cup \{A\}$ is not satisfiable.}
\]
Hence for $(\BC', \PU') := \complete(\I, \BC,\seq{} \apnd [i,A])$, we find $\BC' = \BC$.
\qed
\end{proof}

\begin{lemma}
Each axiom of\/ $\BCL$ is valid.
\end{lemma}
\begin{proof}
We only show some cases. Let 
$\M := (\I, \BC, \seq{}, \val)$ be an initial model.
\begin{enumerate}
\item $\lnot \Box \bot$. By the definition of a model, we have that $\I \cup \set{\BC}$ is satisfiable. Hence $\I \cup \set{\BC} \not\models \bot$, which means $\M \not\models \Box \bot$.
%
%
\item $Qi \to  Qj$ for $i > j$. Assume $\M\models Qi$. That means $i \leq \len(\BC)$. Hence, for $j < i$, we have $j \leq \len(\BC)$, which gives $\M\models Qj$.
\item $ \Acc(i,A) \to ( [i,A]Qi \leftrightarrow \top)$. Assume $\M \models \Acc(i,A)$.  Using Lemma~\ref{l:updates:1}, we get $\M^{[i,A]}\models Qi$. Thus
$\M \models [i,A]Qi \leftrightarrow \top$ as desired.
\item $ \lnot \Acc(i,A) \to( [i,A]Qi \leftrightarrow Qi)$. Assume $\M \not\models \Acc(i,A)$.  We use again Lemma~\ref{l:updates:1} to obtain
$
\M \models [i,A]Qi \leftrightarrow Qi
$.
\item $ [i,A]Qj \leftrightarrow Qj$ for $Qj \in\atomQ$ and $i\neq j$. If $\M \not \not\models \Acc(i,A)$, we obtain \mbox{$\M \models [i,A]Qj \leftrightarrow Qj$} as in the previous case.  If $\M \models \Acc(i,A)$, then again by Lemma~\ref{l:updates:1},
$\M \models  [i,A]Qj \leftrightarrow Qj$ for  $i\neq j$.
\item $ \Acc(i,A) \to ([i,A] \Box B  \leftrightarrow \Box(A \to B))$.
Assume $\M \models \Acc(i,A)$ and let 
\[
(\I, \BC', \PU', \val) := \M^{[i,A]}.
\]
By Lemma~\ref{l:updates:1} we get $\BC' = \BC \apnd A$. Thus $\set{\BC'} = \set{\BC} \cup \{A\}$. By the deduction theorem for classical logic we find
\[
\I \cup \set{\BC} \cup \{A\} \modelsCL B
\quad\text{if and only if}\quad
\I \cup \set{\BC} \modelsCL A \to B,
\]
which yields $\M \models [i,A] \Box B  \leftrightarrow \Box(A \to B)$.
\item $ \lnot \Acc(i,A) \to ([i,A] \Box B  \leftrightarrow \Box B)$. Assume $\M \not\models \Acc(i,A)$.  From Lemma~\ref{l:updates:1}, we  immediately get $\M \models [i,A]\Box B   \leftrightarrow \Box B  $.
\qed
\end{enumerate}
\end{proof}

\begin{lemma}
\label{l:model:1}
Let\/ $\M = (\I, \BC, \PU, \val)$ be an arbitrary model and let\/ $[i,A]$ be a block such that $i > \len(\BC)+1$.
Then we have\/ $\M^{[i,A]} = (\I, \BC, \PU \apnd [i,A] , \val)$.
\end{lemma}
\begin{proof}
Let 
\[
(\BC',\PU') := \complete(\I, \BC,\PU \apnd [i,A]).
\]
Since $\M$ is a model, condition \eqref{eq:model:2} is satisfied. Therefore,
we find that
$\BC' =\BC$ and $\PU' = \PU \apnd [i,A]$, 
which is $\M^{[i,A]} = (\I, \BC, \PU \apnd [i,A] , \val)$.
\qed
\end{proof}

\begin{lemma}
\label{l:model:2}
Let\/ $\M = (\I, \BC, \seq{}, \val)$ be an initial model and let\/ $[i,A]$ be a block such that $i \leq  \len(\BC)+1$.
Then $\M^{[i,A]}$ is an initial model, too.
\end{lemma}
\begin{proof}
Let $\PU = \seq{[i,A]}$ and
\[
(\BC',\PU') := \complete(\I, \BC,\PU).
\]
If $i=\len(\BC)+1$, then $[i,A]$ is removed from $\PU$ in line~5 of Algorithm~\ref{Alg:ChainCompletion}. If $i<\len(\BC)+1$, then $[i,A]$ is removed from $\PU$ in line~13. In both cases we find $\PU' = \seq{}$, which means that $\M^{[i,A]}$ is initial.
\qed
\end{proof}

\begin{lemma}
\label{l:initial:1}
Let $(\I, \BC, \PU, \val)$ be a model and $F$ be an $\langB$-formula such that for each~$[i,A]$  occurring in $F$ we have $i > \len(\BC)+1$.
Then
\[
(\I, \BC, \PU, \val) \models F
\quad\text{if and only if}\quad
(\I, \BC, \seq{}, \val) \models F.
\]
\end{lemma}
\begin{proof}
By induction on the structure of $F$ and a case distinction on the outermost connective.
The only interesting case is $F = [i,A]G$.
Since $i > \len(\BC)+1$ by assumption, we find by Lemma~\ref{l:model:1}  that
$(\I, \BC, \PU, \val)^{[i,A]} = (\I, \BC, \PU \apnd [i,A], \val)$. 
Thus we get 
\begin{equation}\label{eq:normal:1}
(\I, \BC, \PU, \val) \models  [i,A]G \quad\text{if and only if}\quad (\I, \BC, \PU \apnd [i,A], \val) \models G.
\end{equation}
Using I.H.~twice yields
\begin{equation}\label{eq:normal:2}
(\I, \BC, \PU \apnd [i,A], \val) \models G \quad\text{if and only if}\quad (\I, \BC, \seq{[i,A]}, \val) \models G.
\end{equation}
Again since $i > \len(\BC)+1$ we find that
\[
(\I, \BC, \seq{[i,A]}, \val) = (\I, \BC, \seq{}, \val)^{[i,A]}
\]
 and thus
\begin{equation}\label{eq:normal:3}
(\I, \BC, \seq{[i,A]}, \val) \models G \quad\text{if and only if}\quad (\I, \BC, \seq{}, \val) \models  [i,A]G.
\end{equation}
Taking \eqref{eq:normal:1}, \eqref{eq:normal:2}, and \eqref{eq:normal:3} together yields the desired result.
\qed
\end{proof}

Now we can show that the rule  $\mathsf{(SUB)}$ preserves validity.

\begin{lemma}\label{l:subSound:1}
Let $H(P),F,G$ be $\langB$-formulas such that 
$H(F) \leftrightarrow H(G)$ is compliant.
We have that
\[
\text{if $F\leftrightarrow G$ is valid, then $H(F)\leftrightarrow H(G)$ is valid, too.}
\]
\end{lemma}
\begin{proof}
We show the validity of  $H(F)\leftrightarrow H(G)$ by induction on the structure of  $H(P)$.
We distinguish the following cases.
\begin{enumerate}
\item $H$ does not contain $P$. We find $H = H(F) = H(G)$. Hence  $H(F)\leftrightarrow H(G)$ is trivially valid.
\item $H=P$.  We have $H(F)=F$ and $H(G)=G$. Thus $H(F)\leftrightarrow H(G)$ is valid by assumption. 
\item $H = H' \to H''$. Follows immediately by  I.H.
\item $H = \Box H'$ By I.H., we find that $H'(F) \leftrightarrow H'(G)$ is valid. 
Since $\langB$ does not include nested $\Box$-operators, $H'(P)$ is an $\langCl$-formula. Since $H(F)\leftrightarrow H(G)$ is a formula, $F$ and $G$  must be $\langCl$-formulas, too. Hence, $H'(F) \leftrightarrow H'(G)$ is an $\langCl$-formula and we obtain
$\modelsCL H'(F) \leftrightarrow H'(G)$.
Hence we have $\M \models \Box H'(F)$ if and only if $\M \models \Box H'(G)$ for any model~$\M$, which yields that  
$H(F)\leftrightarrow H(G)$ is valid.
\item $H = [i,A]H'$. Let $\M :=(\I, \BC, \seq{}, \val)$ be an initial model. We distinguish the following cases:
	\begin{enumerate}
	\item $i \leq \len(\BC)+1$. By Lemma~\ref{l:model:2},  we find that $\M^{[i,A]}$ is an initial model.  Thus by the I.H.~we infer  $\M^{[i,A]}\models H'(F) \leftrightarrow  H'(G)$, from which we infer 
\[
\M \models  [i,A]H'(F) \leftrightarrow [i,A]H'(G)
\]
by the validity of $\mathsf{(A4)}$. 
	\item $i >  \len(\BC)+1$.  By Lemma~\ref{l:model:1}, we find that  
\[
	\M^{[i,A]} =(\I, \BC, \seq{[i,A]}, \val).
\]
  Since $H(F)$ is compliant, we obtain  that for each~$[j,B]$ occurring in $H(F)$, we have $j > \len(\BC)+1$. 
	Hence we obtain by Lemma~\ref{l:initial:1} that 
	\begin{equation}\label{eq:sub:1}
	\M^{[i,A]} \models H'(F) \quad\text{if and only if}\quad (\I, \BC, \seq{}, \val) \models H'(F).
	\end{equation}
	By I.H.~we get 
	\begin{equation}\label{eq:sub:2}
	 (\I, \BC, \seq{}, \val) \models H'(F) \quad\text{if and only if}\quad (\I, \BC, \seq{}, \val) \models H'(G).
	\end{equation}
	Since $H(G)$ is compliant, we find that $H'(G)$ satisfies the condition of Lemma~\ref{l:initial:1}. Thus we can use that lemma again to obtain
	\begin{equation}\label{eq:sub:3}
	(\I, \BC, \seq{}, \val) \models H'(G) \quad\text{if and only if}\quad \M^{[i,A]} \models H'(G) .
	\end{equation}
	Taking \eqref{eq:sub:1}, \eqref{eq:sub:2}, and \eqref{eq:sub:3} together yields 
\[
\M \models  [i,A]H'(F) \leftrightarrow [i,A]H'(G). 
\]
	\qed
	\end{enumerate}
\end{enumerate}
\end{proof}

We have established that the axioms of $\BCL$ are valid and that $\mathsf{(SUB)}$ preserves validity. 
It is easy to see that the rules $\mathsf{(MP)}$ and $\mathsf{(NEC)}$ also preseve validity. Soundness of $\BCL$ follows immediately.

\begin{corollary}
For each formula $F$ we have
\[
\vdash F \quad \text{implies}\quad \text{$F$ is valid}.
\]
\end{corollary}


\begin{remark}
The reduction axiom $\mathsf{(A3.3)}$ does not hold in non-initial models.
Indeed, let $\M:=(\emptyset, \seq{}, \seq{[2,\top]}, \emptyset)$.
We find that $\M^{[1,P]}=(\emptyset, \seq{P,\top}, \seq{}, \emptyset)$.
Hence $\M^{[1,P]} \models Q2$, which is $\M \models [1,P]Q2$.
But we also have $\M \not\models Q2$.
\end{remark}

\begin{remark}
The above remark also implies that a
 block necessitation rule would not be sound, that is the validity of $F$ does not entail the validity of $[i,A]F$.
Indeed, the axiom $[1,P]Q2 \leftrightarrow Q2$ is valid; but the formula $[2,\top]([1,P]Q2 \leftrightarrow Q2)$ is not valid as shown in the previous remark.
\end{remark}

\begin{remark}
The rule  $\mathsf{(SUB)}$ would not preserve validity if we drop the condition that the conclusion must be compliant.
Indeed, let us again consider the valid formula $[1,P]Q2 \leftrightarrow Q2$.
Without the compliance condition, the rule $\mathsf{(SUB)}$ would derive
$[2,P'][1,P]Q2 \leftrightarrow [2,P']Q2$, which is not a valid formula.
\end{remark}

\section{Normal form}\label{sec:5}

Remember that a formula is compliant if the blockchain updates occur in the correct order. In this section, we establish a normal form theorem for our simple blockchain logic.

\begin{definition}
A \emph{base formula} is a formula that has one of the following forms (which include the case of no blockchain updates):
\begin{enumerate}
\item $[i_1,A_1]\ldots[i_m,A_m] \bot$
\item $[i_1,A_1]\ldots[i_m,A_m] P$ with $P \in \atomP \cup \atomQ$
\item $[i_1,A_1]\ldots[i_m,A_m] \Box B$
\end{enumerate}
Formulas in \emph{normal form} are given as follows:
\begin{enumerate}
\item
each compliant base formula is in normal form
\item
if $F$ and $G$ are in normal form, then so is $F \to G$.
\end{enumerate}
\end{definition}

\begin{remark}
As an immediate consequence of this definition, we obtain that for each formula~$F$,
\[
\text{if $F$ is in normal form, then $F$ is compliant.}
\]
\end{remark}

The following theorem states that for each formula, there is a provably equivalent  formula in normal form.
\begin{theorem}\label{th:equiv:1}
For each $\langB$-formula $F$, there is an $\langB$-formula $G$ in normal form such that\/ $\vdash F \leftrightarrow G$.
\end{theorem}
\begin{proof}
We do an induction on the structure of $F$ and distinguish the following cases:
\begin{enumerate}
\item The cases when $F = \bot$, $F \in \atomP \cup \atomQ$, or $F = \Box B$ are trivial.
\item $F = G \to H$.  By I.H., there are $G'$ and $H'$ in normal form such that $\vdash G \leftrightarrow G'$ and $\vdash H \leftrightarrow H'$. Hence
for $F'  :=G' \to H'$, we find $\vdash F \leftrightarrow F'$ and $F'$ is in normal form.
\item $F = [i_1,A_1]\ldots[i_k,A_k] G$ with $G$ not of the form $[i_{k+1},A_{k+1}] G'$. Subinduction on $G$. We distinguish:
\begin{enumerate}
\item
$G = \bot$,  $G = P \in \atomP \cup \atomQ$, or $G = \Box B$.
In this case, $F$ is a base formula. Using axiom~$\mathsf{(A6)}$, we find a compliant base formula~$F'$ such that $\vdash F \leftrightarrow F'$.
\item
$G = G' \to G''$. Then by axiom~$\mathsf{(A4)}$
\[
\vdash F \leftrightarrow 
([i_1,A_1]\ldots[i_k,A_k] G' \to [i_1,A_1]\ldots[i_k,A_k] G'').
\]
Moreover, by I.H., there are $H'$ and $H''$ in normal form such that 
\[
\vdash H' \leftrightarrow [i_1,A_1]\ldots[i_k,A_k] G' 
\]
and 
\[
\vdash H'' \leftrightarrow [i_1,A_1]\ldots[i_k,A_k] G''. 
\]
We find that $H:=H' \to H''$ is in normal form and $\vdash F \leftrightarrow H$.
\qed
\end{enumerate}
\end{enumerate}
\end{proof}

\section{Completeness}\label{sec:6}

We first show that
$\BCL$ is complete for modal formulas.
The modal language~$\langM$ consists of all update-free $\langB$-formulas. Formally, $\langM$ is given by the following grammar
\[
F ::= \bot \ |\  P \ |\ Q \ |\ F \to F \ |\ \Box A \quad,
\]
where $P \in \atomP$, $Q \in \atomQ$, and $A \in \langCl$.

We need the collection~$\BCL^\Box$ of all $\BCL$ axioms that are given in $\langM$. The usual satisfaction relation for Kripke models is denoted by~$\models_\Box$.

\begin{lemma}\label{l:modalComp:1}
For each $\langM$-formula $F$ we have
\[
\text{$F$ is valid \quad implies \quad $\vdash F$.}
\]
\end{lemma}
\begin{proof}
We show the contrapositive.
Assume $\not \vdash F$. Since $F$ is a modal formula, there is a Kripke model~$\K$ with a world $w$ such that 
\begin{equation}\label{eq:simpleComp:1}
\K,w \not\modelsM F
\end{equation}
and
\begin{equation}\label{eq:simpleComp:3}
\K,w \modelsM G \qquad \text{for all $G \in BCL^\Box$}.
\end{equation}

Based on the Kripke model $\K$, we construct an initial update model $\M = (\I,\BC,\seq{},\val)$ as follows.
Note that because of~\eqref{eq:simpleComp:3}, we have
$\K,w \modelsM Qi \to Qj$ if $j < i$.
Let $k$ be the least $i \in \natNP$ such that $\K,w \not\modelsM Qi$ if it exists and $k:=\omega$ otherwise.
We set:
\begin{enumerate}
\item $\I := \{ A \in \langCl \ |\ \K,w \modelsM \Box A \}$;
\item $\BC := \begin{cases}
	\seq{\top,\ldots,\top} \text{ such that } \len(\BC) = k-1 & \text{if $k<\omega$}\\
	\seq{\top,\top,\ldots} & \text{if $k=\omega$}
\end{cases}$
\item $\val := \{P \in \atomP \ |\  K,w \models P \}$.
\end{enumerate}
This definition of $\BC$ means that $\BC$ is an infinite sequence of $\top$ if $k = \omega$.

For each $\langM$-formula $G$ we have 
\begin{equation}\label{eq:simpleComp:2}
\K,w \modelsM G 
\quad\text{if and only if}\quad
\M \models G.
\end{equation}
We show \eqref{eq:simpleComp:2} by induction on the structure of $G$ and distinguish the following cases:
\begin{enumerate}
\item $G = P \in \atomP$. Immediate by the definition of $\val$.
\item $G = Qi \in \atomQ$.  If $k= \omega$, we have  $\K,w \modelsM Qi$ and, since $\len(\BC) = \omega$, also $\M \models Qi$.
If $k < \omega$, we have 
$\K,w \modelsM Qi$ iff $i \leq k-1 = \len(\BC)$ iff  $\M \models Qi$.
\item $G = \bot$. Trivial.
\item $G = G_1 \to G_2$. By induction hypothesis.
\item $G = \Box A$. If $\K,w \models \Box A$, then $\M \models \Box A$ by the definition of $\I$. If $\M \models \Box A$, then $\I \cup \set{\BC} \models A$. By the definition of $\BC$, this is $\I \models A$. Because $\I$ is deductively closed,  we get $A \in \I$, which yields $\K,w \models \Box A$.
\end{enumerate}
By \eqref{eq:simpleComp:1} and \eqref{eq:simpleComp:2} we conclude
$\M \not\models F$ as desired.
\qed
\end{proof}

We establish completeness for compliant formulas using a translation from compliant formulas to provably equivalent update-free formulas. 
We start with defining a mapping $h$ that eliminates update operators.
\begin{definition}
The mapping $\h$ from $\{[i,A]F \ |\ F \in \langM\}$ to $\langM$ is inductively defined by:
\begin{align*}
\h([i,A] \bot) &:= \bot\\
\h([i,A] P) &:= P \quad\text{for $P \in \atomP$}\\
\h([i,A] Qi) &:= \Acc(i,A) \lor Qi\\
\h([i,A] Qj)&:= Qj \quad\text{for $Qj \in \atomQ$ and $i\neq j$}\\
\h([i,A](F \to G) ) &:= \h([i,A]F)  \to \h([i,A]G) \\
\h([i,A] \Box B) &:= (\Acc(i,A) \land \Box(A \to B) ) \lor (\lnot \Acc(i,A) \land  \Box B)
\end{align*}
\end{definition}

The mapping $\h$ corresponds to the reduction axioms of $\BCL$. Thus it is easy to show the following lemma by induction on the structure of $F$.
\begin{lemma}\label{l:hred:1}
Let $F$ be an $\langB$-formula of the form $[i,A]G$ such that  $G\in \langM$.
We have that $\vdash F \leftrightarrow \h(F)$.
\end{lemma}

We define a translation $\tr$ from $\langB$ to $\langM$
\begin{definition}
The mapping $\tr: \langB \to \langM$ is inductively defined by:
\begin{align*}
\tr(\bot) &:= \bot\\
\tr(P) &:= P \quad\text{for $P \in \atomP \cup \atomQ$}\\
\tr(F \to G) &:= \tr(F)  \to \tr(G) \\
\tr(\Box A) &:=\Box A\\
\tr([i,A]F) &:= \h([i,A] \tr(F))
\end{align*}
\end{definition}

\begin{lemma}\label{l:translation:1}
For each compliant formula $F$, we have
\[
\vdash F \leftrightarrow \tr(F).
\]
\end{lemma}
\begin{proof}
The proof is by induction on the structure of $F$.
There are two interesting cases.
\begin{enumerate}
\item $F = G \to H$. By I.H.~we find $\vdash G \leftrightarrow \tr(G)$ and $\vdash H \leftrightarrow \tr(H)$. 
Thus we have 
\[
\vdash (G \to H) \leftrightarrow (\tr(G) \to \tr(H)), 
\]
which yields the desired result by $\tr(G) \to \tr(H) = \tr(G \to H)$.
\item $F = [i,A]G$. By  I.H.~we find $\vdash G \leftrightarrow \tr(G)$. Since 
$[i,A]G$ is compliant by assumption, we can use $\mathsf{(SUB)}$ to infer $[i,A]G \leftrightarrow[i,A]\tr(G)$. 
By Lemma~\ref{l:hred:1}, we know 
\[
\vdash [i,A]\tr(G) \leftrightarrow \h( [i,A]\tr(G) ).
\]
We finally conclude $ \vdash [i,A]G \leftrightarrow \h( [i,A]\tr(G) )$, which yields the claim since 
\[
\tr([i,A]F) = \h([i,A] \tr(F)).
\]
\qed
\end{enumerate}
\end{proof}

\begin{theorem}\label{th:comp:1}
For each compliant $\langB$-formula $F$ we have
\[
\text{$F$ is valid \quad implies \quad $\vdash F$.}
\]
\end{theorem}
\begin{proof}
Assume that $F$ is a valid and compliant  $\langB$-formula.
By Lemma~\ref{l:translation:1}, we know $\vdash F \leftrightarrow \tr(F)$. Hence by soundness of $\BCL$, we get that $\tr(F)$ is valid, too.
Since $\tr(F)$ is an $\langM$-formula, Lemma~\ref{l:modalComp:1} yields $\vdash \tr(F)$.
Using  Lemma~\ref{l:translation:1} again, we conclude $\vdash F$.
\qed
\end{proof}

Combining Theorem~\ref{th:equiv:1} and Theorem~\ref{th:comp:1} easily yields completeness for the full language. 
\begin{theorem}\
For each  $\langB$-formula $F$ we have
\[
\text{$F$ is valid \quad implies \quad $\vdash F$.}
\]
\end{theorem}
\begin{proof}
Assume $F$ is a $\langB$-formula  that is valid.
By Theorem~\ref{th:equiv:1}, we find a compliant $\langB$-formula $G$ such that 
\begin{equation}\label{eq:fullComp:1}
\vdash F \leftrightarrow G. 
\end{equation}
Hence by soundness of $\BCL$, we know that $G$ is valid, too.
Applying Theorem~\ref{th:comp:1} yields $\vdash G$. 
We finally conclude $\vdash F$ by \eqref{eq:fullComp:1}.
\qed
\end{proof}

\section{Conclusion}\label{sec:7}

We have presented $\BCL$, a dynamic logic to reason about a simple blockchain model. Our semantics does not have the full complexity of the blockchains used in Bitcoin or Ethereum, yet it exhibits two key properties of blockchains: blockchain extensions must preserve consistency and blocks may be received in the wrong order. Note, however, that although receiving blocks in the wrong order is an important logical possibility, it only happens rarely in practice: in the Bitcoin protocol the average generation time of a new block is 10~minutes; the average time until a node receives a block is only 6.5~seconds~\cite{decker2013}.
%
%

In order to illustrate the dynamics of our simple blockchain logic, we state some valid principles of $\BCL$ in the following example.
\begin{example}
The following formulas are valid (and thus provable) in $\BCL$:
\begin{description}
\item[Persistence:] $\Box A \to [i,B]\Box A$. Beliefs are persistent, i.e., receiving a new block cannot lead to a retraction of previous beliefs.
\item[Consistency:] $[i,B]\lnot \Box \bot$. Receiving a new block cannot result in inconsistent beliefs.
\item[Success:] $\Acc(i,A) \to [i,A]\Box A$. If a block $[i,A]$ is acceptable, then $A$ is believed after receiving $[i,A]$.\footnote{We call this prinicple \emph{success}; but it is not related to the notion of a \emph{successful formula} as studied in dynamic epistemic logic, see, e.g.,~\cite{VanDitmarsch2006}.}
\item[Failure:] $(Qi \lor \lnot Q(i-1)) \to ([i,B]\Box A \leftrightarrow \Box A)$.  If the current length of the blockchain is not $i-1$, then receiving a block $[i,B]$ will not change the current beliefs.
\end{description}
\begin{proof}
\begin{enumerate}
\item
Persistence: $\Box A \to [i,B]\Box A$. 
Let $\M := (\I, \BC, \seq{}, \val)$ be an initial model and assume
$\M \models \Box A$. That is $\I  \cup \set{\BC} \models A$.
Let $ (\I, \BC', \PU', \val) := \M^{[i,B]}$.
We find that $\set{\BC} \subseteq \set{\BC'}$. 
Therefore, $\I  \cup \set{\BC'} \models A$, hence $\M^{[i,B]}\models \Box A$ and $\M \models  [i,B]\Box A$.
\item
Consistency: $[i,B]\lnot \Box \bot$.
We let $\M := (\I, \BC, \seq{}, \val)$ be an initial model.
Further, we set 
$
(\I, \BC', \PU', \val) := \M^{[i,B]}
$.
By Lemma~\ref{l:welldef:1} we know that $\I \cup \set{\BC'}$ is satisfiable, i.e., $\I \cup \set{\BC'} \not\models  \bot$.
Hence we have $\M^{[i,B]} \models \lnot \Box \bot$, which is $\M \models  [i,B]  \lnot \Box \bot$.
\item
Success: $\Acc(i,A) \to [i,A]\Box A$.
Let $\M := (\I, \BC, \seq{}, \val)$ be an initial model and assume $\M \models \Acc(i,A)$.
Let  $ (\I, \BC', \PU', \val) := \M^{[i,A]}$.
By
Lemma~\ref{l:updates:1}, we know $\BC' = \BC \apnd A$.
Thus $\I \cup \set{\BC'} \models A$ and, therefore $\M^{[i,A]} \models \Box A$, which is $\M \models  [i,A]  \Box A$.
\item
Failure: $(Qi \lor \lnot Q(i-1)) \to ([i,B]\Box A \leftrightarrow \Box A)$.   
Again, we let $\M := (\I, \BC, \seq{}, \val)$ be an initial model and assume $\M \models  Qi \lor \lnot Q(i-1)$.
We find that $\M \not \models \Acc(i,B)$. Indeed, 
\[
\text{$\M \models  Qi$ implies $\M \not \models \Acc(i,B)$}
\]
and  
\[
\text{$\M \models  \lnot Q(i-1)$ implies $i>1$ and $\M \not \models \Acc(i,B)$.}
\]
Let  $ (\I, \BC', \PU', \val) := \M^{[i,B]}$. By Lemma~\ref{l:updates:1}, we know $\BC' = \BC$.
Therefore, $\M^{[i,B]} \models \Box A$ if and only if $\M\models \Box A$, which yields  $\M \models  [i,B]\Box A \leftrightarrow \Box A$.
\qed
\end{enumerate}
\end{proof}
\end{example}

There  are still many open issues in epistemic blockchain logic. Let us mention three of them. First of all, although blockchains are called \emph{chains}, the data structure that is actually used is more tree-like and there are different options how to choose the valid branch: Bitcoin simply uses the branch that 
has the greastest proof-of-work effort invested in it~\cite{nakamoto} (for simplicity we can think of it as the longest branch);
but recent research shows that the GHOST rule~\cite{Sompolinsky2015} (used, e.g.,~in Ethereum~\cite{Wood2017}) provides better security at higher transaction throughput. We plan to extend $\BCL$ so that it can handle tree-like structures and the corresponding forks of the chain. In particular, this requires some form of probability logic to model the fact that older transactions have smaller probability of being reversed~\cite{DoubleSpend,nakamoto,Sompolinsky2015}.

One of the purposes of blockchains is to provide a data structure that makes it possible to achieve common knowledge among a group of agents in  a distributed system. Logics of common knowledge are 
well-understood~\cite{bs09,FHMV95,jks07,HM95} and we believe that a fully developed blockchain logic should support multiple agents and common knowledge operators.

In a multi-agent setting, each agent (node) has her own instance of a blockchain. Justification logics~\cite{Art01BSL} could provide a formal approach to handle this. Evidence terms could represent blockchain instances and those instances can be seen as justifying the agents' knowledge about the accepted transactions. This approach would require to develop new dynamic justification logics~\cite{BucKuzStu14Realizing,Ren11JLC,KuzStu13LFCS}.
Moreover, if the underlying blockchain model supports forks of the chain, then we need justification logics with probability operators~\cite{KMOS15}.

%


\end{document}